\documentclass[journal]{IEEEtran}
\usepackage{amssymb}
\usepackage{amsmath,amssymb,amsfonts}
\usepackage{multirow}
\usepackage{graphicx}
\usepackage{textcomp}
\usepackage{amsthm}

\usepackage{setspace}

\usepackage{graphicx}
\usepackage{algorithm}
\usepackage{algorithmic}
\usepackage{amsmath}
\usepackage{color}
\usepackage{amsfonts}

\ifCLASSINFOpdf
\else
\fi

\newtheorem{thm}{Theorem}

\newtheorem{defn}{Definition}

\begin{document}
%
\title{Combined Cooling, Heating, and Power System in Blockchain-Enabled Energy Management}
%
%
%

\author{Jianxiong Guo,
	Xingjian Ding,
	Weili Wu,~\IEEEmembership{Member,~IEEE}
	\thanks{J. Guo and W. Wu are with the Department
		of Computer Science, Erik Jonsson School of Engineering and Computer Science, Univerity of Texas at Dallas, Richardson, TX, 75080 USA; X. Ding is with the School of Information, Renmin University of China, Beijing, CHN
		
		E-mail: jianxiong.guo@utdallas.edu}
	\thanks{Manuscript received April 19, 2005; revised August 26, 2015.}}

%
%

\markboth{Journal of \LaTeX\ Class Files,~Vol.~14, No.~8, August~2015}%
{Shell \MakeLowercase{\textit{et al.}}: Bare Demo of IEEEtran.cls for IEEE Journals}
%



\maketitle

\begin{abstract}
The combined cooling, heating and power (CCHP) system is a typical distributed, electricity-gas integrated energy scheme in a community. First, it generates electricity by use of gas, and then exploits the waste heat to supply community with heat and cooling. In this paper, we consider a smart city consisting of a number of communities (CCHPs) and an agent of power grid (APG), where CCHPs can sell energy to the APG according to its bid. To study all utilities of entities in such a city from energy trading, a noncooperative Stackelberg game between APG and CCHPs is formulated. Here, the APG gives a bid for buying the energy from CCHPs, then CCHPs respond to the APG with their optimal energy supply that maximizing their utilities according to this bid. We show that the maximum profit to the APG and utilities to the CCHPs can be obtained at the Stackelberg equilibrium, which is guaranteed to exist and unique. Because the complete information about energy supply of each CCHP is unknown to the APG in advance, we propose a distributed algorithm that is able to find the point of equilibrium through a limited number of iterations. Taking privacy protection and transaction security into consideration, we design a blockchain-enabled energy management system. This system is composed of Internet of Energy (IoE) sub-system and blockchain sub-system, where the information interactions as well as energy transactions between APG and CCHPs can be carried out effectively and safely. Finally, security analysis and numerical simulations show the effectiveness and accuracy of our proposed mechanism.
\end{abstract}


%
\IEEEpeerreviewmaketitle

\section{Introduction}
%
%
%
%
\IEEEPARstart{T}{o} reduce the cost of energy and emission of greenhouse gas, distributed energy system (DES) combined with multiple energy sources have been developed effectively in the last decades. Not only can this energy generated by DES by used by the residents in the local community, but also it can be integrated into the grids. Even though renewable energy \cite{akhtar2015energy}, such as wind energy, hydro energy and solar energy, is gaining popularity, they have not yet become the mainstream methods for supplying energy because of the restrictions on various conditions. Thus, fossil Energy still dominates in the current energy supply. Integrated community energy system (ICES) is an integrated energy system on the user side, which can be considered as a node in a large DES. ICES undertakes the task of satisfying local users' all the energy demands, and can make full use of local natural resources \cite{yu2016brief}. ICES is formed by the coupling of energy supply networks such as medium and low voltage power distribution systems, medium and low voltage natural gas systems, heating systems, and cooling systems \cite{wang2016review}. The joint operation of multiple kinds of energy is of great significance for building an Internet of Energy, improving energy efficiency, and building an environment-friendly society together.

The combined cooling heating and power (CCHP) system is a typical ICES that uses gas turbines or internal combustion engines to generate electricity, heat exchangers to provide heat, and absorption refrigerators to provide cooling. The combined operating efficiency can reach more than 80\%, which is one of the most promising  operating modes in today's integrated energy systems \cite{wu2006combined} \cite{jing2012fuzzy} \cite{jradi2014tri}. One of its advantages is that it can make full use of the local rich natural resources. For example, if coal is produced locally, use coal as fuel; if gas is produced locally, use gas as fuel; wind, hydro and solar energy can do the same thing. Thus, it avoids some of the inherent shortcomings of centralized energy system. In a traditional power grid, the electric energy generated from the centralized node, such as power plant, has to be transmitted by a complex mesh, which results in high losses during transmission and thus low efficiency \cite{matamoros2012microgrids} \cite{wu2015optimal}. Typically, in a community, it is equipped with a CCHP system that meets all energy supply, including cooling, heating, and electricity, for the local residents, and remaining energy can be sold to the power grid or heat/cooling station.

Consider a smart city, it consists of an agent of power grid (APG) responsible for collecting electricity from CCHPs and a number of communities, each of which is equipped with a CCHP system. In this paper, we study the energy trading between APG and CCHPs using noncooperative game theory. As far as we know, this is the first time to consider the energy trading problem in the scenarios of ICES. First, creative utility and profit functions for CCHPs and APG in a city are proposed, which consider both real-world scenarios and mathematical operability. For each CCHP in this city, its utility consists of two parts: one is to serve the residents living in the community it serves by providing them with the energy needed for life, and the other is sold electricity to the APG of its corresponding city for gaining revenues. For the APG in this city, its utility is from the difference in price between market price of electricity and bid for buying electricity from CCHPs belonging to this city. Assume that the production capacity for a day is fixed, a natural problem for each CCHP is how much energy should be sold to the APG. The more it sells, the less it leaves for itself. To encourage the CCHPs to sell more energy, the AGP should offer them a higher price. But like this, it has the potential to make less profit because of the higher cost of buying energy from CCHPs. Thus, this is a dilemma. Then, due to the multilevel decision-making processes of APG and CCHPs in a city, we formulate a Stackelberg game to model the bargain between them, where the AGP is leader and CCHPs are followers. The leader offers a price for buying energy from followers, such that its profit to meet the day's electric load can be maximized. At the same time, the CCHPs respond to the APG with the amount of electricity they are willing to sell given the price offered by APG to maximize their utilities. The properties of this Stackelberg game are analyzed in this paper, and we prove the Stackelberg equilibrium (SE) exists and unique. Because the responses of CCHPs are unpredictable in advance, we propose a distributed algorithm that is guaranteed to reach the unique SE by limited information interactions.

Then, let us imagine a larger area, such as a country, which consists of a number of cities, and each city is covered by an APG responsible for collecting electricity from all communities (CCHPs) in this city. From this, an Internet of Energy (IoE) is formulated. Now, the core of the problem is how to protect privacy in the process of information interaction and how to record transactions safely and effectively. With the rise of blockchain technology, it is widely used as a technical means in energy trading because of its decentralization, security,  and anonymity. Blockchain is a distributed ledger that is able to record transactions in a permanent and verifiable manner. Combined with blockchain technology, we propose a blockchain-enabled energy management system, where energy trading is paid by energy coins. This system can be divided into two sub-systems, called IoE sub-system and blockchain sub-system. The IoE sub-system is made up of the P2P network that connects between APG and CCHPs in a city. It is formed in units of cities, whose main task is to complete the information exchanges and generate transactions between APG and CCHPs within the city. However, the blockchain sub-system is based on the P2P network that connects all cities (APGs) in the whole country. Each APG in the country can be regarded as a node in the blockchain sub-system, which stores complete information about the blockchain distributedly. Thus, the blockchain is established on all APGs in the country to share and verify energy transactions that occurred in this system without third trusted institutions. Here, we can notice that these APGs play the role of both a trader in IoE sub-system and a storer (validator) in blockchain sub-system. They are performed independently by different servers in the APG and do not interfere with each other.

The rest of this paper is organized as follows: Sec. \uppercase\expandafter{\romannumeral2} discusses the-state-of-art work. Sec. \uppercase\expandafter{\romannumeral3} introduces background knowledges about CCHP and defines utility functions. Sec. \uppercase\expandafter{\romannumeral4} presents Stackelberg game and distributed algorithm. Sec. \uppercase\expandafter{\romannumeral5} describes blockchain-enable energy management system in detail. Sec. \uppercase\expandafter{\romannumeral6} analyzes security of our system and conducts numerical simulations. Sec. \uppercase\expandafter{\romannumeral7} is conclusion.

\section{Related Work}
CCHP system and its varieties have been applied widely in many situations to save energy, cost and reduce emissions of greenhouse gas \cite{wu2006combined} \cite{jing2012fuzzy} \cite{jradi2014tri}. Cardona \textit{et al.} \cite{cardona2004validation} explored the economics of CCHP preliminarily, and proposed electric/thermal demand management. The choice between them was determined by a set of complex factors, such as motivation, ability to sell back to the grid and so on. Mago \textit{et al.} \cite{mago2009analysis} constructed a complete framework to analyze and optimize CCHP system following electric load and the thermal load strategies based on their fuel consumption, operation cost, and carbon dioxide emissions. Cho \textit{et al.} \cite{cho2014combined} gave us a comprehensive survey through summarizing more than 170 existing papers, which covers all aspects of CCHP system basically.

Energy management problems about how to integrate DES into a smart grid have been studied intensively. Georgilakis \textit{et al.} \cite{georgilakis2013optimal} summarized the optimally distributed generation placement problem systematically, classified and analyzed current and future research about it. Zhang \textit{et al.}  \cite{zhang2013efficient} considered microgrid as a local energy supplier for domestic buildings by utilizing DES, and stuidied optimal scheduling of energy consumption through mixed-integer programming. Cecati \textit{et al.} \cite{cecati2011smart} exploited DES to make the cost of power delivery minimized by use of an efficient smart grid management system. In addition, the Stackelberg game is suitable for analyzing and designing an energy management system. Maharjan \textit{et al.} \cite{maharjan2013dependable} addressed the demand response management problem by means of establishing a Stackelberg game between multiple utility companies and customers to maximized the profit of each company and utility of each customer. Meng \textit{et al.} \cite{meng2013stackelberg} studied a Stackelberg game between electricity retailer and customers. They adopted genetic algorithms for the retailer to maximize its profit, and developed an analytical solution for customers to minimize their bills. Bu \textit{et al.} \cite{bu2013game} considered a real-time pricing problem for the electricity retailer in the demand-side management, proposed a four-stage Stackelberg game and solved it by a backward induction process. Tushar \textit{et al.} \cite{tushar2014three} proposed an energy management scheme, and formulated a Stackelberg game between residential units and the shared facility controller that can buy energy from residential units or grids. Other researches about game in energy management are shown as \cite{tushar2014prioritizing} \cite{bu2012smart} \cite{asimakopoulou2013leader}.

To deal with the transaction security issues in P2P energy trading, many latest works about it adopt the blockchian technology. Kang \textit{et al.} \cite{kang2017enabling} put forward a localized P2P electricity trading pattern based on consortium blockchain among plug-in hybrid electric vehicles. It was achieved by giving rewards to discharging, thereby balanced the local electricity demand. Li \textit{et al.} \cite{li2017consortium} proposed a P2P energy trading architecture based on consortium blockchain for the industrial Internet of Things. To reduce the transaction confirmation delay, they produced a credit-based payment scheme. Liu \textit{et al.} \cite{liu2018adaptive} designed an adaptive blockchain-based charging scheme to reduce power fluctuation in the grid and cost of electric vehicle users. Aggarwal \textit{et al.} \cite{aggarwal2018energychain} proposed an EnergyChain that permits energy trading between grid and home in a secure manner, including miner choice, transaction verification and block adding. Zhou \textit{et al.} \cite{zhou2019secure} considered the scenario of vehicle-to-grid, and developed a secure energy trading mechanism based on consortium blockchain. Even though that, the design pattern of our blockchain system is different from them.

\begin{figure*}[!t]
	\centering
	\includegraphics[width=\textwidth]{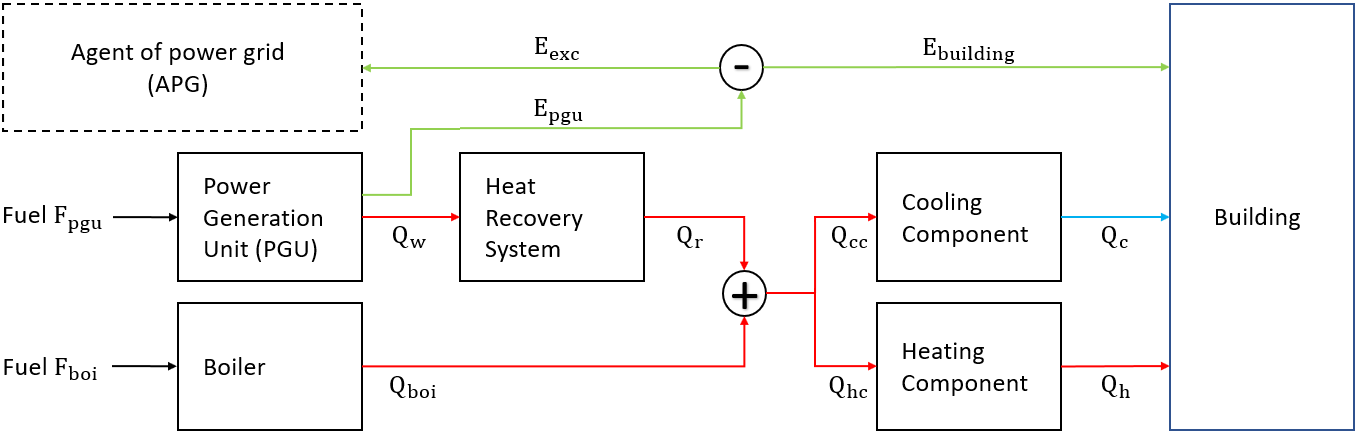}
	\caption{The architecture of combined cooling, heating and power (CCHP) system.}
	\label{fig1}
\end{figure*}

\section{CCHP System model}
The combined cooling, heating and power (CCHP) system is the newest technique of distributed energy utilization. It exploits the recoverable waste heat generated by the electrical power system for the purpose of space cooling, heating, and hot water. In this section, we describe how to model the CCHP system. The architecture of CCHP system is shown in Fig. \ref{fig1}. The fuel, referred to as natural gas in this paper, is inputted into the power generation unit (PGU). It generates electricity that is required to be used in building for lighting, electronic equipment, etc., or sell part of extra electricity to the agent of power grid (APG). Thermal PGU will emit high-temperature smoke gas after generating electricity, called ``waste heat'', which can be used for heating by heating component (heating coil) or cooling by cooling component (chiller). When the waste heat is insufficient for supply to building, it can be replenished by the input of boiler. Then, we talk about the basic properties of CCHP first in this section.

\subsection{Basic Properties of CCHP}
The CCHP system model we describe here follows the electric load, thus the total electricity required by the building can only be supplied by PGU. Thus, we have
\begin{equation}
E_{pgu}=E_{building}+E_{exc}
\end{equation}
Measured in days, the units of $E_{pgu}$, $E_{building}$ and $E_{exc}$ are $\text{J}/\text{day}$. From the PGU, the electric energy demand is equal to electric energy consumption per day. Here, we define a partition coefficient $\beta=E_{building}/E_{pgu}$. The PGU fuel consumption can be defined, that is
\begin{equation}
E_{pgu}=\eta_{pgu}\cdot(q\cdot F_{pgu})
\end{equation}
where $q$ $(\text{J}/\text{m}^3)$ is the calorific value of given fuel and, hence the total thermal energy generated by PGU per day is $q\cdot F_{pgu}$ definitely. For the natural gas, we have $q=3.6\times10^7$ $\text{J}/\text{m}^3$. The unit of gas intake $F_{pgu}$ is $\text{m}^3/\text{day}$. The $\eta_{pgu}$ is the conversion efficiency of PGU, percent energy that transferred from heat to electricity. Given a specific CCHP system, its conversion efficiency is assumed to be a constant.

The waste heat $Q_w$ can be computed by the thermal energy generated by fuel $F_{pgu}$ and electric energy $E_{pgu}$. After passing the heat recovery system, we have $Q_r$ as follows:
\begin{equation}
Q_r=\eta_{rec}\cdot Q_w=\eta_{rec}(1-\eta_{pgu})\cdot q\cdot F_{pgu}
\end{equation}
where the $Q_r$ is the recovered thermal energy, and the units of $Q_w$ and $Q_r$ are $\text{J}/\text{day}$. The $\eta_{rec}$ is the efficiency of heat recovery system. Considering the heat energy supplied by boiler, we have
\begin{equation}
Q_{boi}=\eta_{boi}\cdot(q\cdot F_{boi})
\end{equation}
where $\eta_{boi}$ is the thermal efficiency of boiler. To handle the cooling load, the input of thermal energy $Q_{cc}$ to cooling component can be defined as $Q_c=COP_{cc}\cdot Q_{cc}$, where $Q_c$ $(\text{J}/\text{day})$ is cooling load and $COP_{cc}$ is the coefficient  of performance of chiller. Similarly, to handle the heating load, the input of thermal energy $Q_{hc}$ to heating component can be defined as $Q_h=\eta_{hc}\cdot Q_{hc}$, where $Q_h$ $(\text{J}/\text{day})$ is heading load and $\eta_{hc}$ is the thermal efficiency of coil. By the heat balance, we can know that
\begin{equation}
Q_r+Q_{boi}=Q_{cc}+Q_{hc}
\end{equation}

However, it is complex to determine how to distribute the total thermal energy, $Q_r+Q_{boi}$, to cooling component $Q_{cc}$ and heating component $Q_{hc}$. For example, in summer, the cooling load is significantly heavier than in other seasons due to the hot weather; but in winter, heating demand is higher because need to heat the space of buildings. Apart from this, in different regions, such as tropical or temperate regions, even at different times of the day, the requirements for cooling and heating are different as well. Thus, the allocation of total thermal energy to cooling and heating component depends more on experience. Even though that, we can observe a rule that the peak time of the cooling and heating load. For simplicity, we consider the cooling and heating load as a whole (a control body), and consider its comprehensive thermal efficiency, that is
\begin{equation}
Q_{com}=Q_c+Q_h=\eta_{com}\cdot(Q_r+Q_{boi})
\end{equation}
where $Q_{com}$ $(\text{J}/\text{day})$ is the sum of cooling load and heating load, and $\eta_{com}$ is its comprehensive thermal efficiency, such that $\eta_{com}\in[\min\{COP_{cc},\eta_{hc}\},\max\{COP_{cc},\eta_{hc}\}]$. It can be determined by historic records of cooling and heating load, and we assume it to be a constant. The structure of control body is shown in Fig. \ref{fig2} as follows.

\begin{figure}[h]
	\centering
	\includegraphics[width=\linewidth]{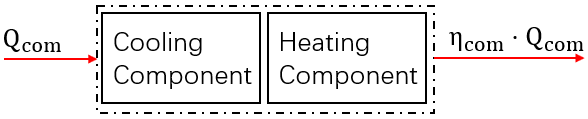}
	\caption{The structure of control body, where we consider cooling/heating components as an entity.}
	\label{fig2}
\end{figure}

Finally, in order to computing cost, the total fuel consumption $F_{tot}$ per day can be defined as
\begin{equation}
F_{tot}=F_{pgu}+F_{boi}
\end{equation}

\subsection{Distributed System Model}
Consider in a city, there are a number of separate communities, and in each community, there is a CCHP associated with it and responsible for supplying cooling, heat, and electricity to these residents living in this community. In this city, it exists an APG, which collects extra electric energy from communities (CCHPs) belonging to this city. For one thing, each CCHP addresses all energy supplies in its community. A CCHP can be a single, or several CCHPs integrated by an aggregate, which can be considered as a single CCHP entity. It can decide the amount of fuel to be consumed to generate thermal and electric energy, and how much electric energy is sold to the APG for making revenue. For another thing, the APG attempts to buy electricity from these CCHPs in this city at a lower price, transfers it to other places where electricity is needed, and sell it at a normal electricity price.

First, it assumes there is a set $\mathcal{I}$ containing all CCHPs in a city. Each CCHP $i\in\mathcal{I}$ consumes fuel $F^i_{pgu}$ and $F^i_{boi}$ independently, and it is able to manage its partition coefficient $\beta^i$ that sells $(1-\beta^i)\cdot E^i_{pgu}$ electric energy to the grid for making revenue. To reduce pollution, the government usually adopts some policies to restrict the gas usage for each community. For each CCHP $i\in\mathcal{I}$, we say its gas usage restriction is given by $F^i_{tot}$. If it exceeds $F^i_{tot}$, this community will face serious fines and bad credit records. Thus, we can consider $F^i_{tot}$ as a predefined constant, which indicates how much fuel can be consumed at most in a day. Then, we denote by $p_b$ the buying price of APG from CCHPs and $p_s$ the selling (retail) price of APG, where the APG is willing to pay $p_b$ per unit of energy to each CCHPs for purchasing their electricity. Here, we define a partition coefficient $\alpha^i=F^i_{pgu}/F^i_{tot}$. Based on that, we propose the utility function for each CCHP $i\in\mathcal{I}$ obtained from its $\alpha^i$ and $\beta^i$, that is
\begin{flalign}
U^i&=k^i_1\ln(1+b^i_1qF^i_{tot}\eta^i_{pgu}\alpha^i\beta^i)\nonumber\\
&+k^i_2\ln(1+b^i_2((qF^i_{tot}\eta_{rec}^i(1-\eta_{pgu}^i)\alpha^i+qF^i_{tot}\eta^i_{boi}\nonumber\\
&(1-\alpha^i))\eta_{com}))\nonumber\\
&+p_b(qF^i_{tot}\eta^i_{pgu}\alpha^i(1-\beta^i))
\end{flalign}
where $k^i_1$, $b^i_1$, $k^i_2$ and $b^i_2$ are adjustable parameters. The utility function $U^i(\alpha^i,\beta^i)$ is composed of three items. In Equation (8), the first item is the utility of electric energy generated by CCHP $i$ that used for electricity demand in its community. The second item is the heat energy generated by CCHP $i$ that used for cooling and heating demand in its community, and it is composed of waste heat from PGU and boiler. The third item is the revenue from the electric energy that sold to the APG at the price of $p_b$.

Unfortunately, the utility $U^i(\alpha^i,\beta^i)$ is a multivariate function, and cannot be guaranteed to be a concave function because the first item involves two variables. This causes great difficulty for our follow-up mathematical analysis processing. Generally, the electric load of a community is much larger than its heating and cooling load, because electricity is a more advanced energy source, and can be converted to other energy easily. Given total fuel $F^i_{tot}$, it is valid to think the waste heat $Q^i_w$ from EPG, $qF^i_{tot}\cdot (1-\eta_{pgu}^i)\cdot\eta_{rec}^i\cdot\eta^i_{com}$, is enough to supply heat energy for the residents in community $i$. Thus, all the fuel should be inputted into PGU, namely we can set $\alpha^i=1$. Like this, utility function $U^i(\alpha^i,\beta^i)$ can be converted to univariate function $U^i(\beta^i)$, that is
\begin{flalign}
U^i&=k^i_1\ln(1+b^i_1qF^i_{tot}\eta^i_{pgu}\beta^i)\nonumber\\
&+k^i_2\ln(1+b^i_2qF^i_{tot}(1-\eta_{pgu}^i)\eta_{rec}^i\eta_{com}^i)\nonumber\\
&+p_b(qF^i_{tot}\eta^i_{pgu}(1-\beta^i))
\end{flalign}
where $\beta^i\in[0,1]$. Consider the item, $\ln(1+x)\in[0,+\infty)$ when $x\geq0$, in related researches, they usually adopted function $k^i\ln(1+x)$ to estimate the utility for community \cite{tushar2014three} \cite{samadi2010optimal} \cite{tushar2012economics}. But it exists a drawback to this utility function that logarithmic function does not have asymptote. Thus, we add a parameter $b^i_1$ such that $\ln(1+b^i_1x)\in[0,+\infty)$, and we can use it to model the the utility of electric energy used in community. This parameter $b^i_1$ to control the variation range of term $\ln(1+b^i_1x)$, to avoid grow infinitely and make each parameter meaningful. Our objective is to let $\ln(1+b^i_1x)=1$ if $\beta^i=1$, at this time, the utility of the first item of (9) reaches the maximum value. Thus, we have
\begin{equation}
b^i_1=\frac{e-1}{qF^i_{tot}\eta^i_{pgu}}
\end{equation}
and for heat energy
\begin{equation}
b^i_2=\frac{e-1}{qF^i_{tot}(1-\eta_{pgu}^i)\eta_{rec}^i\eta_{com}^i}
\end{equation}

On the other hand, consider the APG, it has no power of pricing, because retail price is usually regulated by the government. Thus, the retail price of electricity of APG can be considered as a constant. To make revenue, the APG wishes to collect electric energy from those distributed CCHPs at a price as low as possible. Here, we assume the APG is willing to buy electricity at the price $p_b$ which aims to maximize its earning. Based on Equation (9), if $p_b$ is low, the CCHPs tend to increase their partition coefficient $\beta$, namely sell less, and use more electricity to serve themselves. They consume more energy to improve the quality of life, or simply reduce power generation, results in the revenue of APG is reduced. In contrast, if $p_b$ is high, even approached to retail price $p_s$, the profit per unit of electricity will be small. Therefore, it is important for the APG to set a valid bid $p_b$, not only encourage CCHPs to sell more electricity, but also ensure that the APG has sufficient profitability. Usually, $p_b\in(p_m,p_c)$, where $p_c$ is the cost price of generating electricity by the APG. The APG needs to generate electricity by itself at a cost price $p_c$ when the electric energy collected from CCHPs is not enough to meet the requirement. Then, we propose the utility function to capture the profit gained by the AGP, that is
\begin{flalign}
L&=(p_s-p_b)\sum_{i\in\mathcal{I}}\left(qF^i_{tot}\eta^i_{pgu}(1-\beta^i)\right)\nonumber\\
&+(p_s-p_c)\left(R-\sum_{i\in\mathcal{I}}\left(qF^i_{tot}\eta^i_{pgu}(1-\beta^i)\right)\right)
\end{flalign}
where $R$ is the electric load of APG, and we assume $R\geq\sum_{i\in\mathcal{I}}qF^i_{tot}\eta^i_{pgu}$. The utility function $L(p_b)$ is composed of two items. In Equation (12), the first item is the profit of selling the electric energy bought from all CCHPs. The second item is the profit of selling the electric energy generated by itself, which is constrained by the total demand $R$.

\section{Trading between APG and CCHPs}
In the last section, we have introduced the utility functions of APG and CCHPs, which depend on partition coefficient $\beta^i$ for each $i\in\mathcal{I}$ and bid $p_b$ of APG. In the process of energy trading, they need to determine $\{\beta^i\}_{i\in\mathcal{I}}$ and $p_b$ respectively. The APG offer an appropriate price $p_b$ to purchase the electric energy generated by each CCHP $i\in\mathcal{I}$, in order to maximize its profit $L$. Each CCHP $i\in\mathcal{I}$ responds to APG with the amount of electricity that is willing to sell according to the bid of APG by adjusting parameter $\beta_i$ to maximize its utility $U^i$. In this section, we describe the process of trading and formulate their optimal strategies step by step.

\subsection{Objective of APG and CCHPs}
Given a bid $p_b$ by APG, each CCHP $i\in\mathcal{I}$ hopes to sell part of electric energy to APG for making revenue by controlling the partition coefficient $\{\beta_i\}_{i\in\mathcal{I}}$. Thus, the objective function for each CCHP $i\in\mathcal{I}$ can be defined, that is
\begin{equation}
\max_{\beta^i}U^i(\beta^i) \text{ s.t., } \beta^i\in[0,1]
\end{equation}
Then, its first-order derivative is
\begin{equation}
\frac{\partial U^i}{\partial\beta^i}=qF^i_{tot}\eta^i_{pgu}\cdot\left(\frac{k^i_1b^i_1}{1+b^i_1qF^i_{tot}\eta^i_{pgu}\beta^i}-p_b\right)
\end{equation}
The maximum utility value for each $i\in\mathcal{I}$ can be obtained by solving its first-order defferential condition $\partial U^i/\partial\beta^i=0$, so we have
\begin{equation}
\beta^i_\Diamond=\frac{1}{qF^i_{tot}\eta^i_{pgu}}\cdot\left(\frac{k^i_1}{p_b}-\frac{1}{b^i_1}\right)
\end{equation}
which shows that the response strategy of each CCHP to the bid of APG. Here, we need to note that the choice of parameter $k^i_1$ must be in a valid range such that $\beta^i_\Diamond\in[0,1]$ for any value of bid $p_b\in[p_m,p_c]$. According to Equation (10) and (15), we can get the range of $k^i_1$ as follows:
\begin{equation}
p_c\cdot\left(\frac{qF^i_{tot}\eta^i_{pgu}}{e-1}\right)\leq k^i_1\leq p_m\cdot\left(\frac{eqF^i_{tot}\eta^i_{pgu}}{e-1}\right)
\end{equation}
where it assumes $p_c\leq e\cdot p_m$, or else no such $k^i_1$ can keep $\beta^i_\Diamond\in[0,1]$ satisfied. Form (15), the partition coefficient $\beta^i_\Diamond$ determined by each $i\in\mathcal{I}$ is proportional invervely to the bid $p_b$ of APG. Therefore, the CCHP $i$ tends to sell more electric energy by decreasing $\beta^i$ for a higher buying price.

Conversely, let us look at the side of APG, which aims to maximize its profit by buying electric energy from all CCHPs at a reasonable price. Thus, the objective function of APG can be defined, that is
\begin{equation}
\max_{p_b}L(p_b) \text{ s.t., } p_b\in[p_m,p_c]
\end{equation}
Consider the relationship between $\{\beta_i\}_{i\in\mathcal{I}}$ and $p_b$ from (15), substitute (15) into (12), we have
\begin{flalign}
L&=(p_c-p_b)\sum_{i\in\mathcal{I}}\left(qF^i_{tot}\eta^i_{pgu}-\left(\frac{k^i_1}{p_b}-\frac{1}{b^i_1}\right)\right)\nonumber\\
&+(p_s-p_c)R
\end{flalign}
Then, is first-order derivative is
\begin{equation}
\frac{\partial L}{\partial p_b}=\sum_{i\in\mathcal{I}}\left(\frac{k^i_1p_c}{p_b^2}-\left(qF^i_{tot}\eta^i_{pgu}+\frac{1}{b^i_1}\right)\right)
\end{equation}
The maximum profit value for the APG can be obtained by solving its first-order defferential condition $\partial L/\partial p_b=0$, so we have
\begin{equation}
p_b^\circ=\sqrt{\frac{p_c\sum_{i\in\mathcal{I}}k^i_1}{\sum_{i\in\mathcal{I}}(qF^i_{tot}\eta^i_{pgu}+(b^i_1)^{-1})}}
\end{equation}
which shows that the optimal price $p_b^\circ$ is affected by the number of CCHPs and their properties. From (20), we can observe that the optimal price is interfered by the cost price $p_c$ as well. With the increase of $p_c$, the optimal price should increase theoretically. In order to make the profit maximized, the APG should set its bid according to (20) to collect electricity from CCHPs, but from (17), $p_b\in[p_m,p_c]$, the optimal strategy of APG can be shown as follow:
\begin{equation}
p_b^\Diamond=\left
\{\begin{IEEEeqnarraybox}[\relax][c]{l's}
p_c,&if $p_b^\circ\geq p_c$\\
p_m,&if $p_b^\circ\leq p_m$\\
p_b^\circ,&if $p_m<p_b^\circ<p_c$%
\end{IEEEeqnarraybox}
\right.
\end{equation}
Because the profit function is concave, which will be proved later, $L(p_c)$ is the maximum value when $p_b^\circ\geq p_c$; Similarly, $L(p_m)$ is the maximum value when $p_b^\circ\leq p_m$.

From (20) and (21), the APG can obtain the optimized bid $p_b^\Diamond$ that maximizing its profit function, shown as (18), easily if it can acquire complete information about those parameters, such as $k^i_1$, $b^i_1$, $F^i_{tot}$ and $\eta^i_{pgu}$, for each CCHP $i\in\mathcal{I}$. However in the real situation, it seems unrealistic that complete information about parameter setting of all CCHPs can be accessed by the APG in a direct way because of its flexibility or out of privacy protection. Thus, a noncooperative Stackelberg game is formulated between APG and CCHPs to decide on the variable $\{\beta_i\}_{i\in\mathcal{I}}$ and $p_b$.

\subsection{Noncooperative Stackelberg Game}
A noncooperative Stackelberg game generally refers to the multilevel decision making processes of a number of independent decision-makers in response to the decision taken by the leading player of the game \cite{doi:10.1137/1.9781611971132}. In this section, we formulate a noncooperative Stackelberg game to model the bargain between APG and CCHPs, where the APG is the leader, and the CCHPs are the followers. The game $\mathbb{G}$ is formally defined by its strategic form as
\begin{equation}
\mathbb{G}=\left\{\text{APG}\cup\{i\}_{i\in\mathcal{I}},L,\{U^i\}_{i\in\mathcal{I}},p_b,\{\beta_i\}_{i\in\mathcal{I}}\right\}
\end{equation}
where the bid $p_b$ is offered by APG (leader) first, and each CCHP $i\in\mathcal{I}$ (follower) responds to the APG with its strategy $\beta^i$. Then, $L$ is the profit function of APG, shown as (12), and $U^i$ is the utility function of CCHP $i$, shown as (9). As said before, the purpose of APG and CCHPs are to maximize its profit in (12) and their utilities in (9) by adapting their corresponding trading strategies. The optimal solution of this game can be obtained when the APG can get the maximized profit at a bid $p^*_b$ given the CCHPs' best responses. In other words, none of them, including the leader and followers, can get a larger profit and utilities through altering their strategies unilaterally. At this time, the Stackelberg equilibrium (SE) is formulated, which is defined as follows:
\begin{defn}[Stackelberg Equilibrium]
Given a Stackelberg game $\mathbb{G}$ defined in (22), we say a set of strategies $\left(p_b^*,\{\beta^i_*\}_{i\in\mathcal{I}}\right)$ reaches an Stackelberg equilibrium of game $\mathbb{G}$ if and only if the following inequalities are met,
\begin{flalign}
&U^i\left(p^*_b,\beta^i_*,\{\beta^j_*\}_{j\in\mathcal{I}\backslash\{i\}}\right)\geq U^i\left(p^*_b,\beta^i,\{\beta^j_*\}_{j\in\mathcal{I}\backslash\{i\}}\right)\\
&L\left(p^*_b,\{\beta^i_*\}_{i\in\mathcal{I}}\right)\geq L\left(p_b,\{\beta^i_*\}_{i\in\mathcal{I}}\right)
\end{flalign}
where this strategy $\left(p_b^*,\{\beta^i_*\}_{i\in\mathcal{I}}\right)$ is the point of equilibrium and satisfies $p_b^*\in[p_m,p_c]$, $\beta^i\in[0,1]$ for $i\in\mathcal{I}$.
\end{defn}

Based on Definition 1, neither the leader and the followers can improve their utilities by changing their strategies respectively when the SE $\left(p_b^*,\{\beta^i_*\}_{i\in\mathcal{I}}\right)$ is reached. However, it is possible for the noncooperative game with pure strategies that the point of equilibrium does not exist \cite{doi:10.1137/1.9781611971132}. Hence, we want to know whether our proposed game $\mathbb{G}$ exists an SE.
\begin{thm}
Between APG and CCHPs in $\mathcal{I}$, it exists a unique SE in our Stackelberg game $\mathbb{G}$.
\end{thm}
\begin{proof}
For each CCHP $i\in\mathcal{I}$, and based on the first-order derivative of $U^i$, shown as (14), we have
\begin{equation}
\frac{\partial^2 U^i}{\partial\beta^{i2}}=-k^i_1\cdot\left(\frac{b^i_1qF^i_{tot}\eta^i_{pgu}}{1+b^i_1qF^i_{tot}\eta^i_{pgu}\beta^i}\right)^2
\end{equation}
where $\partial^2 U^i/\partial\beta^{i2}<0$, thus utility function $U^i$, shown as (9), is strictly concave with respect to $\beta^i$ for $i\in\mathcal{I}$. Thus, given any bid $p_b$ by APG, for each CCHP $i\in\mathcal{I}$, it exists a unique partition coefficient $\beta^i_\Diamond$ selected from $[0,1]$, shown as (15), that maximizing $i$'s utility function. A SE can be reached when the APG and CCHPs have their maximum profit and utilities respectively. Because the strategy $\beta^i_\Diamond$ is unique for $i\in\mathcal{I}$ given a bid $p_b$, we have that the game $\mathbb{G}$ reaches an equilibrium only if the APG is capable of finding the best bid $p_b^*$ at which each CCHP $i$ select the unique partition $\beta^i_*$ that maximizing its utility given this best bid $p_b^*$ simultaneously. Given the optimized responses by all CCHPs, consider (18) and (19), we have
\begin{equation}
\frac{\partial^2 L}{\partial p_b^2}=-\frac{2p_c}{p_b^3}\cdot\sum_{i\in\mathcal{I}}k^i_1
\end{equation}
where $\partial^2 L/\partial p_b^2<0$, thus profit function, shown as (18), is strictly concave with respect to the bid $p_b$. The APG can acquire the maximum profit by finding a unique bid $p_b^*$ given the optimized response of  CCHPS. Therefore, our proposed game $\mathbb{G}$ exists a unique SE definitely.
\end{proof}

\subsection{Distributed Algorithm}
As mentioned earlier, complete information about CCHPs is not available for the APG. Hence, instead of centralized fashion, a distributed algorithm needs to be designed, where the APG is not required to know the parameters information of CCHPs, but only receive the amount of electric energy each of them plans to sell. The APG and all of CCHPs can reach the unique SE of Stackelberg game $\mathbb{G}$ in an iterative manner by use of limited communications between leader and followers. This distributed algorithm is shown in Algorithm \ref{a1}.

\begin{algorithm}[h]
	\caption{\text{Finding SE}}\label{a1}
	\begin{algorithmic}[1]
		\STATE Initialize: $p^*_b:=p_m$, $L^*:=(p_s-p_c)R$
		\FOR {each CCHP $i\in\mathcal{I}$}
		\STATE Initialize: $\beta^i_*=0$
		\ENDFOR
		\FOR {the bid $p_b$ from $p_m$ to $p_c$}
		\FOR {each CCHP $i\in\mathcal{I}$}
		\STATE CCHP $i$ decides on its partition coefficient $\beta^i_\Diamond$ according to $\beta^i_\Diamond=\frac{1}{qF^i_{tot}\eta^i_{pgu}}\cdot\left(\frac{k^i_1}{p_b}-\frac{1}{b^i_1}\right)$
		\ENDFOR
		\STATE The APG computes its profit $L$ based on the response $\beta^i_\Diamond$ for each $i\in\mathcal{I}$ according to
		\begin{equation*}
			L=(p_c-p_b)\sum_{i\in\mathcal{I}}\left(qF^i_{tot}\eta^i_{pgu}(1-\beta^i_\Diamond)\right)+(p_s-p_c)R
		\end{equation*}
		\IF {$L\geq L^*$}
		\STATE $p^*_b:=p_b$, $L^*:=L$
		\STATE $\beta^i_*:=\beta^i_\Diamond$ for each $i\in\mathcal{I}$
		\ENDIF
		\ENDFOR
		\RETURN $\left(p_b^*,\{\beta^i_*\}_{i\in\mathcal{I}}\right)$
	\end{algorithmic}
\end{algorithm}
\noindent
For Algorithm \ref{a1}, in each iteration, the APG offers a bid first, then each CCHP $i\in\mathcal{I}$ decides on its best partition coefficient $\beta^i_\Diamond$ according to (15) based on the bid $p_b$ offered by APG, and sends it the APG. After receiving all responses $\{\beta^i_\Diamond\}_{i\in\mathcal{I}}$ from CCHPs, the APG is able to compute its current profit $L$ according to (12). If this profit $L$ is better than before, it updates the global variables $p^*_b$ and $L^*$. The result $\left(p_b^*,\{\beta^i_*\}_{i\in\mathcal{I}}\right)$ returned by Algorithm \ref{a1} sastisfies the definition in (23) (24), thus the game $\mathbb{G}$ reaches the SE.

\section{Blockchain-Enabled Energy Management}
Mentioned before, in a city, it can be divided into a number of communities, each of which is equipped with a CCHP system responsible for supplying energy for this community. In this city, there is an APG, and it can trade with all CCHPs in this city. Like this, an ecosystem of energy trading is formulated, which is composed of a number of cities. For instance, a country is a typical energy ecosystem. Here, the ecosystem is denoted by $\mathbb{E}=\{\text{E}_1,\text{E}_2,\cdots\}$, where we have $\text{E}_j=\{\{\text{APG}_j\},\{\text{CCHP}_j^1,\text{CCHP}_j^2,\cdots\}\}$ for each $\text{E}_j\in\mathbb{E}$. In this section, we propose the blockchain-enable energy management system, shown in Fig. \ref{fig3}, which is made up of IoE sub-systems and blockchain sub-system. Generally speaking, a city $\text{E}_j\in\mathbb{E}$ is an IoE sub-system, which is composed of an APG and a number of CCHPs in this city. It mainly realizes information interactions and finishes energy transactions between APG and CCHPs in this city. Participants of the IoE sub-system are connected with each other by P2P communication. The blockchain sub-system consists of all APGs in the ecosystem $\mathbb{E}$, where all APGs are connected by P2P communication. Apart from the effect described above in the IoE sub-system, it needs to verify and record those energy transactions between APG and CCHPs in a secure and trusted manner. The workflow of our proposed blockchain-enabled energy management system is as follows: First, a transaction between APG and CCHPs is initiated and finished in an IoE sub-system; Then this transaction can be verified and stored permanently in the blockchain sub-system.

\begin{figure}[!t]
	\centering
	\includegraphics[width=\linewidth]{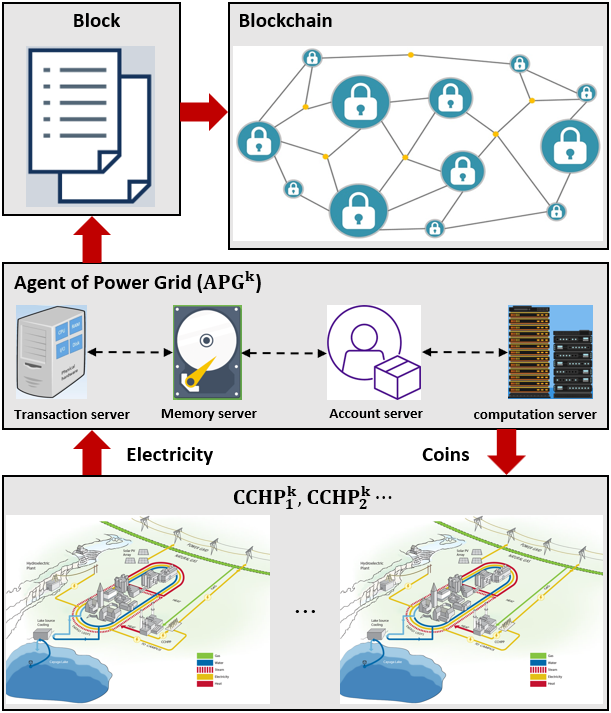}
	\caption{The workflow of our blockchained-enable energy management system, and the schematic diagram of CCPH is from \cite{cchp.org}.}
	\label{fig3}
\end{figure}

\subsection{IoE Sub-system}
Consider in Fig. \ref{fig3}, for each $\text{APG}_j\in\text{E}_j\in\mathbb{E}$ in the ecosystem, there are four major components: a transaction server, an account server, a memory server, and a computation server. The transaction server is a central controller, which is mainly responsible for giving a bid, collecting the responses from CCHPs in its city, adjusting trading strategy, and decide whether to trade. Each entity, including APG and CCHPs, in each city $\text{E}_j\in\mathbb{E}$ has a personal account in the account server of $\text{APG}_j$, which stores personal transaction records. Besides, there is a wallet associated with each personal account, and the digital assets of each entity are stored in its wallet in the form of energy coin \cite{kang2017enabling}. For privacy protection, true address of the wallet is hidden by a public key (random pseudonym), and the mapping relationships between the personal accounts and the public keys of their associated wallets are stored in this account server. Besides, the memory server and computation server are mainly used in the later blockchain sub-system, we will introduce in the next section.

As introduced before, the IoE sub-system is mainly responsible for information interactions and finishing transactions. The information includes the bid of APG and responses (amount of energy) of CCHPs. Thus, our core problem is how to protect the privacy during information interaction and how to ensure security during trading. Here, we design a smart contract to solve this problem. Smart contract empower credible transactions without third parties based on blockchain technology, which can ensure transactions to be trackable and irreversible, but reduce time and cost at the same time. In the beginning, each entity, including APG and CCHPs, needs to register on trusted institution, e.g., a department authorized by government, to become a legitimate entity, and obtain its ID, public key, private key, certificate and wallet address, which is assigned by this trusted institution. Here, the public and private key can be achieved by some existing algorithms, such as elliptic curve digital signature algorithm \cite{johnson2001elliptic}, lattice-based signature scheme \cite{gao2018secure} and anti-quantum signature scheme \cite{yin2018anti}. The certificate links each entity with its registration information. The information of entities forms a mapping, and stored in the account server of the APG in its city. After initialization, they can verify the integrity of its wallet in corresponding account server. Now, details of smart contract for energy trading is described as follows:
\begin{enumerate}
	\item \textit{Energy requesting:} The APG (transaction server) broadcast the energy request to all of the CCHPs in its city. If it makes energy requesting, it needs to send the message that includes its bid $p_b$.
	\item \textit{Response:} Each CCHP in this city that is willing to sell its energy to the APG has to determine the amount of energy to sell, and give responses back to the APG with it. If no response, means that it does not trade.
	\item \textit{Decide on whether start a transaction:} After receiving all responses from CCHPs, the APG (transaction server) need to decide whether to start trading or change bid. If changing bid, go back to step (1) and give a new buying price; or else go to the next step.
	\item \textit{Start trading:} The APG (transaction server) broadcasts the notification that it agrees to start this trade to all of CCHPs in its city, and CCHPs send their wallet addresses to the APG. Then, the energy is delivered from CCHPs to their corresponding APG by electricity pipeline or wireless transmission equipment.
	\item \textit{Payment by energy coins:} The energy coins are transferred to the wallet of CCHPs that have completed providing electric energy according to their wallet addresses. Then, these transaction records are generated by the APG, needed to be verified and digitally signed by the CCHPs to ensure its integrity and authenticity.
	\item \textit{Transactions recording:} The information of transactions finished in this city has to be recorded in account server first, and then upload into the memory server at the same time. These recorded transactions in memory server are broadcasted to other APGs (memory servers) in the whole ecosystem to be further verified and recorded. Invalid transactions will be neglected, and others will be added into the new block in future.
\end{enumerate}

In addition, it is worth noting that all the information that transmitted in this system, including energy requesting, response, transaction, and payment, must be signed by the sender in order to make sure the authenticity and traceability of the information and permit others to validate it. Here, for each CCHP, its private key is used for signing transactions, and its public key is shared with all APGs in the ecosystem, where they can use it to verify its signature. 

\subsection{Bloackchain Sub-system}
After a transaction is uploaded into the memory server, it will be broadcasted to all the memory servers of APGs in the ecosystem. All of the transaction records (blocks) in the ecosystem $\mathbb{E}$ are stored in the memory server. Next, all memory servers start competing for accounting rights, e.g., adding the new block to blockchain. The mining process is similar to find a correct proof-of-work in Bitcoin system. Denoted $data$ by the information of the header of previous block, merkel root of transactions, timestamps and so on, this process required the memory server of each APG in the ecosystem to look for a nonce $a$ such that ${Hash}(data+a)<b$ \cite{alqassem2014towards}, where $b$ is the level of difficulty to find $a$ that controlled by the system. If an APG finds such a nonce successfully, it must broadcast its block and this nonce to other APGs (memory servers) as soon as possible. Other APGs audit it and decide to accept or reject it after receiving the result. If the number of APGs in the ecosystem that agree to accept it satisfies the condition of consensus, this new block will be added into blockchain. At this time, the energy coins associated with transactions in this new added block is transferred to corresponding wallets substantially. Besides, the APG that generating this new block will be rewarded by a certain number of energy coins.

Shown as above, APGs require computational power to compete for adding a new block, thus, in each APG, there is a computation server responsible for providing enough computational power. If someone wants to create a false block, it will finish a correct proof-of-work independently (find a valid nonce) before other APGs and has ability to dominate the majority of APGs. It's almost impossible for a single node. If someone wants to modify a transaction in some block, it will be impeded because the header of its subsequent block is related to this transaction. Thus, a tiny modification to a transaction in some block will affect all the subsequent blocks.

To execute the consensus process, it is executed among all APGs in the ecosystem, and the leader is the first one that gives a correct proof-of-work. The leader broadcasts its proof-of-work, block, and timestamp to all other APGs, then these APGs verify this proof-of-work as well as block, and broadcast their verified block each other with their signature. Each APG contrasts its result with those verified results from other APGs, and reports this comparison result to the leader with its signature. After receiving all responses from APGs, the leader will perform on them with statistical analysis. If all of them approve this block, the leader will broadcast a notification that claims the new block should be added into blockchain in chronological order, and it is rewarded by energy coins. If some APGs reject this block, the leader will find out the reasons why they reject, and send the block to them for second verification if necessary. The total time that used to reach consensus is almost unchanged no matter what the network size is, when the number of APGs in the ecosystem remains stable \cite{luu2016secure}. All the processes are completed by the memory server of each APG.

\section{Security Analysis and Simulations}
In this section, we will analyze the security of our proposed blockchain-enabled energy management system, and conduct simulations to verify the Stackelberg game.
\subsection{Security Analysis}
For each APG in the ecosystem, it is not only an entity that participating in the transaction in IoE sub-system, but also a node that stores the blockchain in blockchain sub-system. Different roles are played by different servers, and these servers work together but work independently. Thus, our energy management system inherits the characteristics of the blockchain, shown as follows:

\begin{enumerate}
	\item \textit{Decentralization:} The energy trades between APG and CCHPs are carried out in a P2P manner, and the bookkeeping process is finished among APGs without a third trusted intermediary.
	\item \textit{Privacy protection:} Each CCHP uses its public key to communicate with the APG in its city, and this public key is shared among the APGs to be verified without disclosing true identity. Besides, the wallets of CCHPs are hidden by pseudonyms and can only be accessed by corresponding key and certificate, which avoiding malicious attacks against a specific entity.
	\item \textit{Authentication:} All transactions need to be audited and verified publicly in the consensus process by all APGs in the ecosystem. It is extremely hard to dominate the majority of APGs to create an unreal block because of the difficulty of proof-of-work.
	\item \textit{Integrity:} Any block that new added into blockchain contains the hash of the previous block, and its subsequent block contains its hash. A malicious attacker that attempts to modify a transaction must create a new chain after the block this transaction is in by dominating the majority of computational power, this is impossible. Besides, every transaction in block is encrypted, it is hard to be decrypted without the private key.
	\item \textit{Transparency:} The nature of decentralization requires the blockchain to be saved in all memory servers of APGs. Thus, it is transparent to every entity, and CCHPs are able to check and confirm those transactions that they participate in easily.
	\item \textit{No double-spending:} The blockchain provides all entities with a public ledger of transactions in the ecosystem, which avoids double-spending potentially.
\end{enumerate}

\subsection{Numerical Simulation}
First of all, we plan to evaluate the properties of utility function of APG and CCHPs. Let us consider such a city that there is an APG and only one CCHP. Typically, the calorific value of natural gas is $3.6\times10^7$ $\text{J}/\text{m}^3$ at the standard atmosphere and the measure for electricity is $\text{ kW}\cdot\text{h}$, where $1\text{ kW}\cdot\text{h}=3.6\times10^6\text{ J}$. According to the latest U.S. electricity price, that is $0.2\text{ dollar}/\text{kW}\cdot\text{h}$, hence, we can set $p_s=5.5\times10^{-8}\text{ dollar}/\text{J}$. In our energy management system, it can be considered as $p_s=5.5\times10^{-8}\text{ coin}/\text{J}$ equivalently. We assume the cost price $p_c=4\times10^{-8}\text{ coin}/\text{J}$ because the cost price should be less than retail price. Besides, for this CCHP, we assume its $F_{tot}=200\text{ m}^3$, $\eta_{pgu}=1$ and $R=0$ for simplicity, because these settings does not affect the properties of our objective functions. From (16), we have $p_c\leq e\cdot p_m$, thus we are able to assume $p_m=2\times10^{-8}\text{ coin}/\text{J}$ and its parameter $k_1$ satisfies $k_1\in[167.6093,227.8046]$. Thereby we have the range of variables, that is $\beta\in[0,1]$ for CCHP and $p_b\in[2\times10^{-8},4\times10^{-8}]$ for APG definitely.

Fig. \ref{fig4} draws the relationship between CCHP's utility $U$ and partition coefficient $\beta$ under different bid $p_b$ given by APG, where $k_1=197.7069$. Shown in Fig. \ref{fig4}, as $\beta$ increases, these utility functions increase first and then decrease regardless of what the bid is. It proves that the utility function $U$, shown as (9), is concave. Furthermore, as $p_b$ increases, the point of maximum utility moves to the right, which means that the CCHP tends to allocate more electric energy that sold to the APG by increasing $\beta$ in order to obtain the maximum utility. Fig. \ref{fig5} draws the relationship between APG's profit $L$ and buying price $p_b$ under different parameter $k_1$ set by CCHP. Shown in Fig. \ref{fig5}, these profit functions increase first and then decrease regardless of what the $k_1$ is. It proves that the profit function $L$, shown as (18), is concave. Furthermore, as $k_1$ increases, the point of maximum profit moves to the right, which means that the APG has to give a higher bid to buy the electric energy from CCHP in order to obtain the maximum profit. Here, a larger $k_1$ implies the energy that used to serve community contributes much to the total utility, thus the APG has to offer a higher bid to buy energy.

\begin{figure}[!t]
	\centering
	\includegraphics[width=\linewidth]{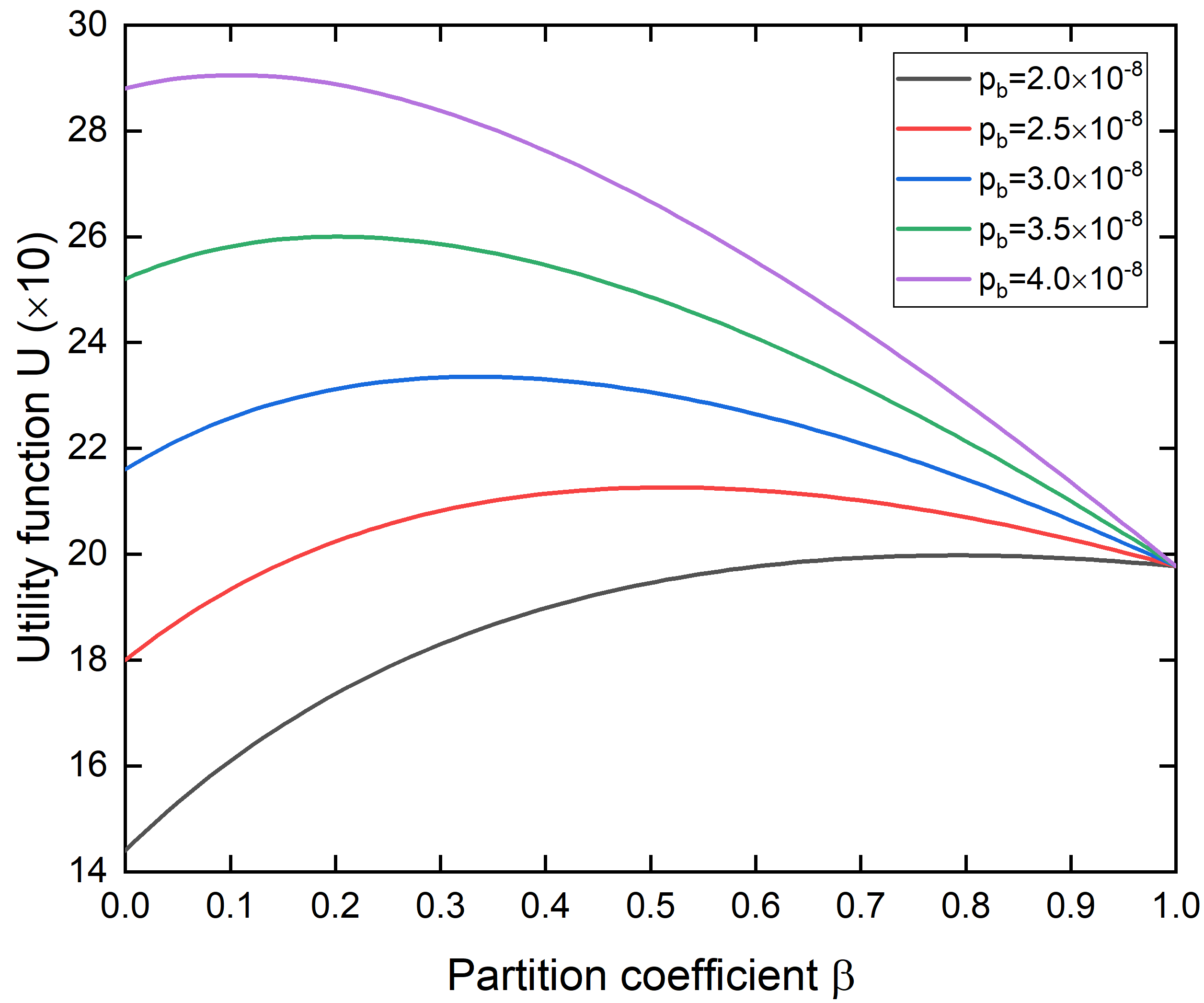}
	\caption{The relationship between CCHP's utility $U$ and partition coefficient $\beta$ under different bid $p_b$ given by APG.}
	\label{fig4}
\end{figure}

\begin{figure}[!t]
	\centering
	\includegraphics[width=\linewidth]{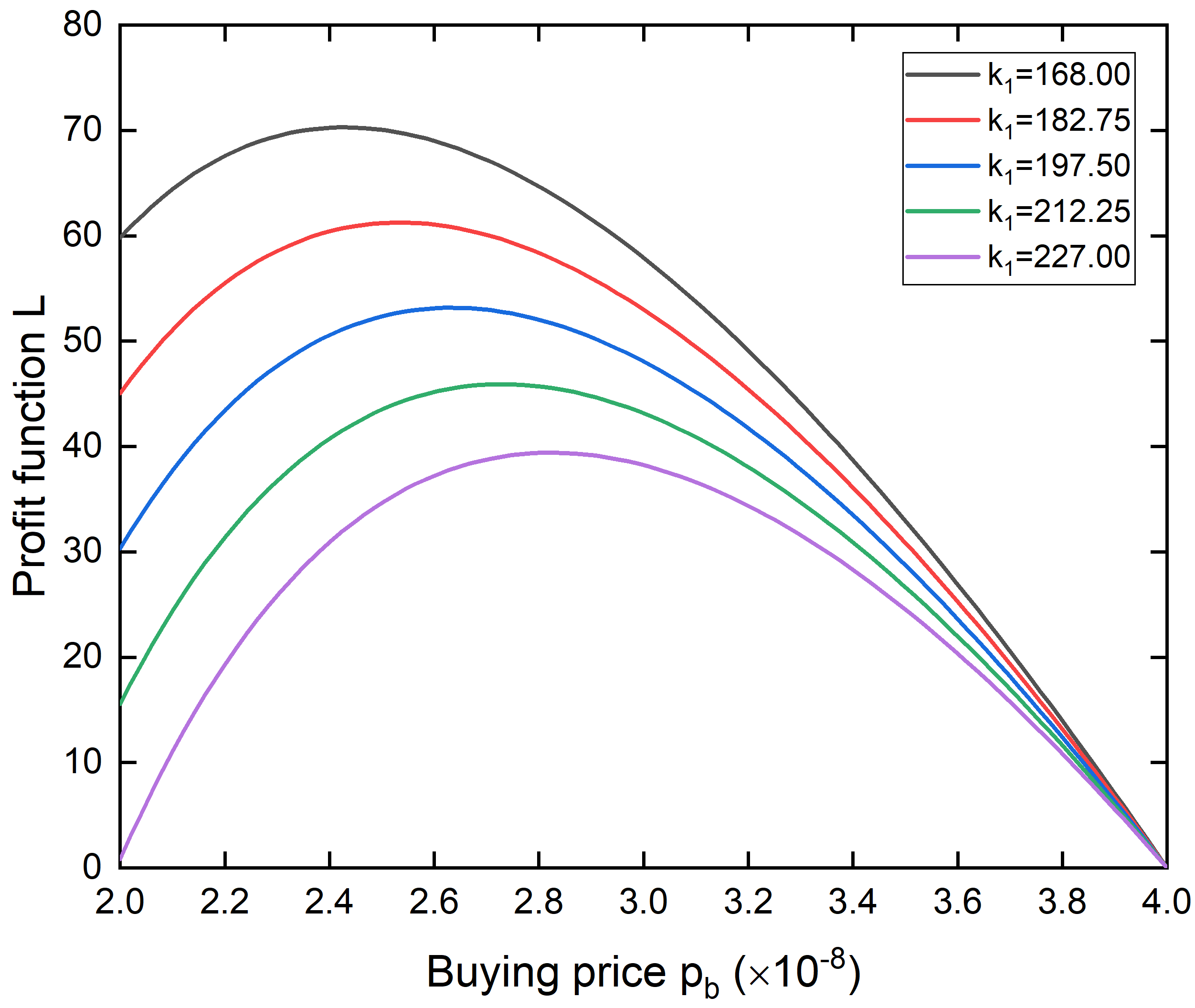}
	\caption{The relationship between APG's profit $L$ and buying price $p_b$ under different parameter $k_1$ set by CCHP.}
	\label{fig5}
\end{figure}

Then, we consider a common city there are a number of CCHPs, denoted by a set $\mathcal{I}$, in this city want to sell electric energy to the APG. The setting of $p_c$ and $p_m$ is the same as above, thus we have $p_b\in[2\times10^{-8},4\times10^{-8}]$ as well. For each CCHP $i\in\mathcal{I}$, we assume its $F^i_{tot}=200\text{ m}^3$ and $\eta^i_{pgu}=1$ for simplicity. The parameter $k^i_1$ for different CCHP $i\in\mathcal{I}$ can be specified arbitrarily, but all of them satisfies $k^i_1\in[167.6093,227,8046]$ definitely. Here, we set electric load $R$ in (18) as $R=2\cdot\sum_{i\in\mathcal{I}}qF^i_{tot}\eta^i_{pgu}$. Now, we can evalute the performance of the convergence to the SE by following distributed Algorithm \ref{a1}.

Fig. \ref{fig6} and Fig. \ref{fig7} draw the process of APG's buying price $p_b$ and CCHPs' partition coefficient $\beta^i$, $i\in\mathcal{I}$, converged to SE by following Algorithm \ref{a1} in a city with five CCHPs set as above. At the beginning, bid $p_b$ offer by APG is low, the CCHPs are unwilling to sell their electric energy to the APG, hence, their partition coefficient is very high. By interacting with each CCHP in this city, the APG adjusts its strategy (increases its bid) gradually in each iteration to encourage CCHPs to sell more energy in order to obtain larger profit. After the $32$-th iterations, the APG gains the largest profit at the buying price $2.64\times10^{-8}$ $\text{coin}/\text{J}$, which is the point of Stackelberg equilibrium, thus their SE is reached.

\begin{figure}[!t]
	\centering
	\includegraphics[width=\linewidth]{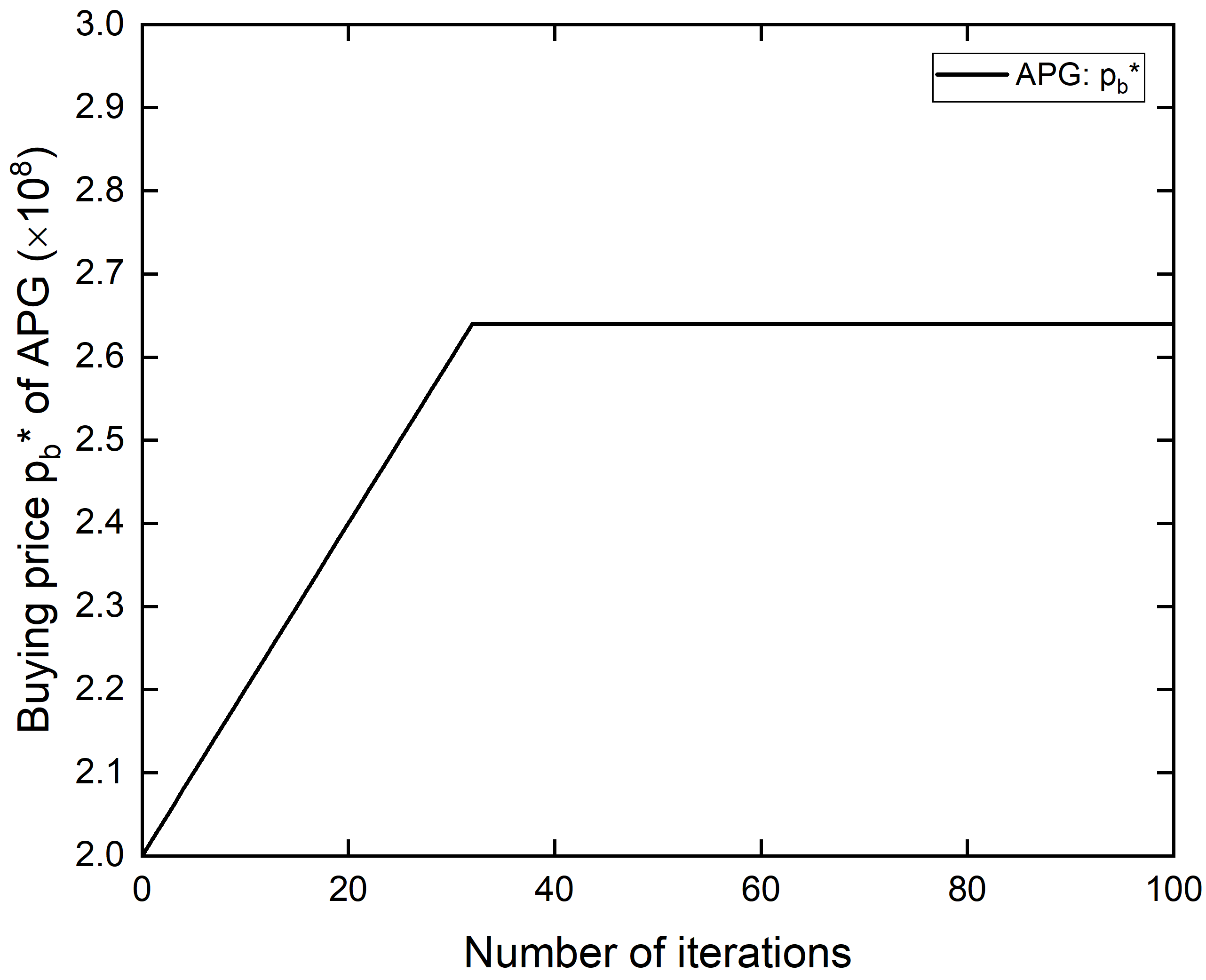}
	\caption{The process of the APG's buying price $p_b$ converged to SE by following Algorithm \ref{a1}.}
	\label{fig6}
\end{figure}

\begin{figure}[!t]
	\centering
	\includegraphics[width=\linewidth]{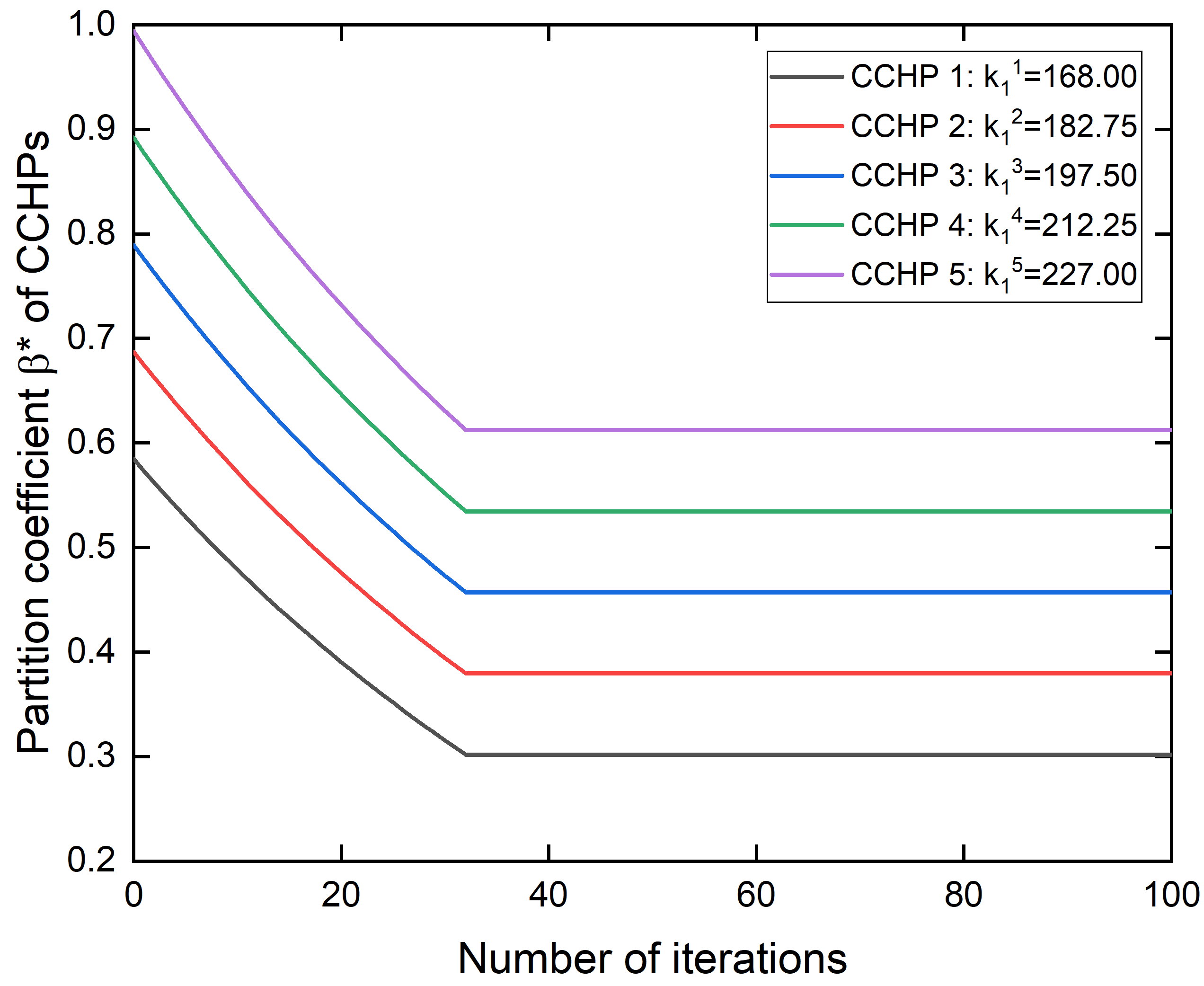}
	\caption{The process of those CCHPs' partition coefficient $\beta^i$ converged to SE by following Algorithm \ref{a1}.}
	\label{fig7}
\end{figure}

Finally, we have discussed before that the APG can acquire the optimal bid directly in a centralized manner if all parameters of CCHPs in its city are known. Here, we can compare the performance of profits that are obtained in both the centralized and our proposed distributed manner. The comparison results are shown in Table \ref{table1}. Here, we set electric load $R$ as $R=30\cdot qF_{tot}$, where $F^i_{tot}=F_{tot}=200\text{ m}^3$, $i\in\{1,2,\cdots,30\}$, and $k_1^i$, $i\in\{1,2,\cdots,30\}$, is sampled uniformly from $[167.6093,227,8046]$. Shown in Table \ref{table1}, the profits of APG at the SE of this game obtained by following our distributed algorithm are very close to that computed in the centralized manner regardless of the number of CCHPs. The profit in a centralized manner is slightly higher than that under distributed algorithm, thus its performance is better because of complete information. We execute 100 iterations between $[2\times10^{-8},4\times10^{-8}]$, thus the stride is $2\times10^{-10}$. To improve the accuracy of distributed algorithm further, we can reduce the stride through increasing the number of iterations. In addition, we assess how the profit of APG changes with the different number of CCHPs in the city by comparing with the base profit. The base profit is computed under the the circumstance that there is no CCHP in the city, which is equal to $(p_s-p_c)\cdot R=3240$. This implies that all required energy $R$ has to be generated at the cost price. Accordingly, the base profit is lower than that involved with CCHPs. Next, the profit of APG increases gradually with the increase of the number of CCHPs in the city, because the APG is able to buy more electric energy from CCHPs at a price lower than cost price. Thus, the profit will be increased certainly. The increment (the last column) in Table \ref{table1} measures the performance compared to base profit because of CCHPs' existence, namely quantified by actual profit divided by base profit. The effect is getting more and more significant that increasing from $108\%$ to $147\%$ as the number of CCHPs increases.

\begin{table}[h]
	\renewcommand{\arraystretch}{1.3}
	\caption{The APG's profits obtained in centralized and distributed manner under different \# of CCHPs in the city.}
	\label{table1}
	\centering
	\begin{tabular}{|c|c|c|c|c|c|}
		\hline
		\bfseries $|\mathcal{I}|$& \multicolumn{2}{|c|}{\bfseries Centralized} & \multicolumn{2}{|c|}{\bfseries Distributed} & \bfseries Incret\\
		\hline
		--- & \bfseries $p_b^*$ & Profit & $p_b^*$ & Profit & -------\\
		\hline
		5  & 2.6300 & 3507.2231 & 2.6400 & 3507.2026 & 108\%\\
		\hline
		10 & 2.5969 & 3800.6279 & 2.6000 & 3800.6236 & 117\%\\
		\hline
		15 & 2.6225 & 4050.5182 & 2.6200 & 4050.5142 & 125\%\\
		\hline
		20 & 2.6645 & 4255.7724 & 2.6600 & 4255.7551 & 131\%\\
		\hline
		25 & 2.6656 & 4507.6399 & 2.6600 & 4507.6065 & 139\%\\
		\hline
		30 & 2.6487 & 4799.9251 & 2.6400 & 4799.8273 & 148\%\\
		\hline
	\end{tabular}
\end{table}

\section{Conclusion}
In this paper, we constructed CCHP system, and discussed the energy trading between APG and CCHPs in a smart city, where a Stackelberg game was formulated. Then, we designed utility functions through valid simplifications in order to be suitable for completing this game. We showed that the SE between APG and CCHPs in the city is guaranteed to exist and unique. Thus, we proposed a distributed algorithm that is able to reach the SE by limited iterations. To protect privacy and ensure transaction security, we created a blockchain-enabled energy management system, where the information interactions and energy transactions can be completed without a trusted third institution. Security analysis showed our system is secure and reliable. The results of numerical simulations indicated that our model is valid, and SE can be reached by our proposed distributed algorithm.

In the future, this energy trading can be improved and extended further. CCHP system is an integrated energy system, which implies not only selling electricity to the grids, but supplying heat to heat stations or cooling to cooling stations. Hence, it forms a multi-leader Stackelberg game, where the utility functions of the followers are multivariate functions.


%

\section*{Acknowledgment}

This work is partly supported by National Science Foundation under grant 1747818.

\ifCLASSOPTIONcaptionsoff
  \newpage
\fi



%

\bibliographystyle{IEEEtran}
\bibliography{references}

\begin{thebibliography}{10}
\providecommand{\url}[1]{#1}
\csname url@samestyle\endcsname
\providecommand{\newblock}{\relax}
\providecommand{\bibinfo}[2]{#2}
\providecommand{\BIBentrySTDinterwordspacing}{\spaceskip=0pt\relax}
\providecommand{\BIBentryALTinterwordstretchfactor}{4}
\providecommand{\BIBentryALTinterwordspacing}{\spaceskip=\fontdimen2\font plus
\BIBentryALTinterwordstretchfactor\fontdimen3\font minus
  \fontdimen4\font\relax}
\providecommand{\BIBforeignlanguage}[2]{{%
\expandafter\ifx\csname l@#1\endcsname\relax
\typeout{** WARNING: IEEEtran.bst: No hyphenation pattern has been}%
\typeout{** loaded for the language `#1'. Using the pattern for}%
\typeout{** the default language instead.}%
\else
\language=\csname l@#1\endcsname
\fi
#2}}
\providecommand{\BIBdecl}{\relax}
\BIBdecl

\bibitem{akhtar2015energy}
F.~Akhtar and M.~H. Rehmani, ``Energy replenishment using renewable and
  traditional energy resources for sustainable wireless sensor networks: A
  review,'' \emph{Renewable and Sustainable Energy Reviews}, vol.~45, pp.
  769--784, 2015.

\bibitem{yu2016brief}
X.~Yu, X.~Xu, S.~Chen, J.~Wu, and H.~Jia, ``A brief review to integrated energy
  system and energy internet,'' \emph{Transactions of China Electrotechnical
  Society}, vol.~31, no.~1, pp. 1--13, 2016.

\bibitem{wang2016review}
W.~Wang, D.~Wang, H.~Jia, Z.~Chen, B.~Guo, H.~Zhou, and M.~Fan, ``Review of
  steady-state analysis of typical regional integrated energy system under the
  background of energy internet,'' \emph{Proceedings of the CSEE}, vol.~36,
  no.~12, pp. 3292--3305, 2016.

\bibitem{wu2006combined}
D.~Wu and R.~Wang, ``Combined cooling, heating and power: A review,''
  \emph{progress in energy and combustion science}, vol.~32, no. 5-6, pp.
  459--495, 2006.

\bibitem{jing2012fuzzy}
Y.-Y. Jing, H.~Bai, and J.-J. Wang, ``A fuzzy multi-criteria decision-making
  model for cchp systems driven by different energy sources,'' \emph{Energy
  Policy}, vol.~42, pp. 286--296, 2012.

\bibitem{jradi2014tri}
M.~Jradi and S.~Riffat, ``Tri-generation systems: Energy policies, prime
  movers, cooling technologies, configurations and operation strategies,''
  \emph{Renewable and Sustainable Energy Reviews}, vol.~32, pp. 396--415, 2014.

\bibitem{matamoros2012microgrids}
J.~Matamoros, D.~Gregoratti, and M.~Dohler, ``Microgrids energy trading in
  islanding mode,'' in \emph{2012 IEEE Third International Conference on Smart
  Grid Communications (SmartGridComm)}.\hskip 1em plus 0.5em minus 0.4em\relax
  IEEE, 2012, pp. 49--54.

\bibitem{wu2015optimal}
Y.~Wu, X.~Tan, L.~Qian, D.~H. Tsang, W.-Z. Song, and L.~Yu, ``Optimal pricing
  and energy scheduling for hybrid energy trading market in future smart
  grid,'' \emph{Ieee transactions on industrial informatics}, vol.~11, no.~6,
  pp. 1585--1596, 2015.

\bibitem{cardona2004validation}
E.~Cardona and A.~Piacentino, ``A validation methodology for a combined heating
  cooling and power (chcp) pilot plant,'' \emph{J. Energy Resour. Technol.},
  vol. 126, no.~4, pp. 285--292, 2004.

\bibitem{mago2009analysis}
P.~Mago and L.~Chamra, ``Analysis and optimization of cchp systems based on
  energy, economical, and environmental considerations,'' \emph{Energy and
  Buildings}, vol.~41, no.~10, pp. 1099--1106, 2009.

\bibitem{cho2014combined}
H.~Cho, A.~D. Smith, and P.~Mago, ``Combined cooling, heating and power: A
  review of performance improvement and optimization,'' \emph{Applied Energy},
  vol. 136, pp. 168--185, 2014.

\bibitem{georgilakis2013optimal}
P.~S. Georgilakis and N.~D. Hatziargyriou, ``Optimal distributed generation
  placement in power distribution networks: models, methods, and future
  research,'' \emph{IEEE Transactions on power systems}, vol.~28, no.~3, pp.
  3420--3428, 2013.

\bibitem{zhang2013efficient}
D.~Zhang, N.~Shah, and L.~G. Papageorgiou, ``Efficient energy consumption and
  operation management in a smart building with microgrid,'' \emph{Energy
  Conversion and Management}, vol.~74, pp. 209--222, 2013.

\bibitem{cecati2011smart}
C.~Cecati, C.~Citro, A.~Piccolo, and P.~Siano, ``Smart operation of wind
  turbines and diesel generators according to economic criteria,'' \emph{IEEE
  Transactions on Industrial Electronics}, vol.~58, no.~10, pp. 4514--4525,
  2011.

\bibitem{maharjan2013dependable}
S.~Maharjan, Q.~Zhu, Y.~Zhang, S.~Gjessing, and T.~Basar, ``Dependable demand
  response management in the smart grid: A stackelberg game approach,''
  \emph{IEEE Transactions on Smart Grid}, vol.~4, no.~1, pp. 120--132, 2013.

\bibitem{meng2013stackelberg}
F.-L. Meng and X.-J. Zeng, ``A stackelberg game-theoretic approach to optimal
  real-time pricing for the smart grid,'' \emph{Soft Computing}, vol.~17,
  no.~12, pp. 2365--2380, 2013.

\bibitem{bu2013game}
S.~Bu and F.~R. Yu, ``A game-theoretical scheme in the smart grid with
  demand-side management: Towards a smart cyber-physical power
  infrastructure,'' \emph{IEEE Transactions on Emerging Topics in Computing},
  vol.~1, no.~1, pp. 22--32, 2013.

\bibitem{tushar2014three}
W.~Tushar, B.~Chai, C.~Yuen, D.~B. Smith, K.~L. Wood, Z.~Yang, and H.~V. Poor,
  ``Three-party energy management with distributed energy resources in smart
  grid,'' \emph{IEEE Transactions on Industrial Electronics}, vol.~62, no.~4,
  pp. 2487--2498, 2014.

\bibitem{tushar2014prioritizing}
W.~Tushar, J.~A. Zhang, D.~B. Smith, H.~V. Poor, and S.~Thi{\'e}baux,
  ``Prioritizing consumers in smart grid: A game theoretic approach,''
  \emph{IEEE Transactions on Smart Grid}, vol.~5, no.~3, pp. 1429--1438, 2014.

\bibitem{bu2012smart}
S.~Bu, F.~R. Yu, Y.~Cai, and X.~P. Liu, ``When the smart grid meets
  energy-efficient communications: Green wireless cellular networks powered by
  the smart grid,'' \emph{IEEE Transactions on Wireless Communications},
  vol.~11, no.~8, pp. 3014--3024, 2012.

\bibitem{asimakopoulou2013leader}
G.~E. Asimakopoulou, A.~L. Dimeas, and N.~D. Hatziargyriou, ``Leader-follower
  strategies for energy management of multi-microgrids,'' \emph{IEEE
  transactions on smart grid}, vol.~4, no.~4, pp. 1909--1916, 2013.

\bibitem{kang2017enabling}
J.~Kang, R.~Yu, X.~Huang, S.~Maharjan, Y.~Zhang, and E.~Hossain, ``Enabling
  localized peer-to-peer electricity trading among plug-in hybrid electric
  vehicles using consortium blockchains,'' \emph{IEEE Transactions on
  Industrial Informatics}, vol.~13, no.~6, pp. 3154--3164, 2017.

\bibitem{li2017consortium}
Z.~Li, J.~Kang, R.~Yu, D.~Ye, Q.~Deng, and Y.~Zhang, ``Consortium blockchain
  for secure energy trading in industrial internet of things,'' \emph{IEEE
  transactions on industrial informatics}, vol.~14, no.~8, pp. 3690--3700,
  2017.

\bibitem{liu2018adaptive}
C.~Liu, K.~K. Chai, X.~Zhang, E.~T. Lau, and Y.~Chen, ``Adaptive
  blockchain-based electric vehicle participation scheme in smart grid
  platform,'' \emph{IEEE Access}, vol.~6, pp. 25\,657--25\,665, 2018.

\bibitem{aggarwal2018energychain}
S.~Aggarwal, R.~Chaudhary, G.~S. Aujla, A.~Jindal, A.~Dua, and N.~Kumar,
  ``Energychain: Enabling energy trading for smart homes using blockchains in
  smart grid ecosystem,'' in \emph{Proceedings of the 1st ACM MobiHoc Workshop
  on Networking and Cybersecurity for Smart Cities}, 2018, pp. 1--6.

\bibitem{zhou2019secure}
Z.~Zhou, B.~Wang, M.~Dong, and K.~Ota, ``Secure and efficient vehicle-to-grid
  energy trading in cyber physical systems: Integration of blockchain and edge
  computing,'' \emph{IEEE Transactions on Systems, Man, and Cybernetics:
  Systems}, vol.~50, no.~1, pp. 43--57, 2019.

\bibitem{samadi2010optimal}
P.~Samadi, A.-H. Mohsenian-Rad, R.~Schober, V.~W. Wong, and J.~Jatskevich,
  ``Optimal real-time pricing algorithm based on utility maximization for smart
  grid,'' in \emph{2010 First IEEE International Conference on Smart Grid
  Communications}.\hskip 1em plus 0.5em minus 0.4em\relax IEEE, 2010, pp.
  415--420.

\bibitem{tushar2012economics}
W.~Tushar, W.~Saad, H.~V. Poor, and D.~B. Smith, ``Economics of electric
  vehicle charging: A game theoretic approach,'' \emph{IEEE Transactions on
  Smart Grid}, vol.~3, no.~4, pp. 1767--1778, 2012.

\bibitem{doi:10.1137/1.9781611971132}
T.~Başar and G.~J. Olsder, \emph{Dynamic Noncooperative Game Theory, 2nd
  Edition}.\hskip 1em plus 0.5em minus 0.4em\relax Society for Industrial and
  Applied Mathematics, 1998.

\bibitem{cchp.org}
``Campus energy systems,''
  \url{https://sustainablecampus.cornell.edu/campus-initiatives/buildings-energy/campus-energy}.

\bibitem{johnson2001elliptic}
D.~Johnson, A.~Menezes, and S.~Vanstone, ``The elliptic curve digital signature
  algorithm (ecdsa),'' \emph{International journal of information security},
  vol.~1, no.~1, pp. 36--63, 2001.

\bibitem{gao2018secure}
Y.-L. Gao, X.-B. Chen, Y.-L. Chen, Y.~Sun, X.-X. Niu, and Y.-X. Yang, ``A
  secure cryptocurrency scheme based on post-quantum blockchain,'' \emph{IEEE
  Access}, vol.~6, pp. 27\,205--27\,213, 2018.

\bibitem{yin2018anti}
W.~Yin, Q.~Wen, W.~Li, H.~Zhang, and Z.~Jin, ``An anti-quantum transaction
  authentication approach in blockchain,'' \emph{IEEE Access}, vol.~6, pp.
  5393--5401, 2018.

\bibitem{alqassem2014towards}
I.~Alqassem and D.~Svetinovic, ``Towards reference architecture for
  cryptocurrencies: Bitcoin architectural analysis,'' in \emph{2014 IEEE
  International Conference on Internet of Things (iThings), and IEEE Green
  Computing and Communications (GreenCom) and IEEE Cyber, Physical and Social
  Computing (CPSCom)}.\hskip 1em plus 0.5em minus 0.4em\relax IEEE, 2014, pp.
  436--443.

\bibitem{luu2016secure}
L.~Luu, V.~Narayanan, C.~Zheng, K.~Baweja, S.~Gilbert, and P.~Saxena, ``A
  secure sharding protocol for open blockchains,'' in \emph{Proceedings of the
  2016 ACM SIGSAC Conference on Computer and Communications Security}, 2016,
  pp. 17--30.

\end{thebibliography}

%




\end{document}